\documentclass[a4paper]{article}
\usepackage{graphicx}
\usepackage{amsmath}  
\usepackage{microtype}
\usepackage{amsfonts}

\usepackage[margin=1in,letterpaper]{geometry}

\usepackage{amsthm}

\newtheorem{theorem}{Theorem}
\newtheorem{lemmar}[theorem]{Lemma}
\newtheorem{definitionr}[theorem]{Definition}
\newtheorem{conjecturer}[theorem]{Conjecture}

\newtheorem{propositionr}[theorem]{Proposition}

\newtheorem{exampler}[theorem]{Example}

\newcommand{\rb}[1]{\left( #1 \right)}
\newcommand{\ceil}[1]{\left\lceil #1 \right\rceil}

\usepackage{cite} 
\usepackage[final]{hyperref} 
\hypersetup{
	colorlinks=true,       
	linkcolor=red,        
	citecolor=blue,        
	filecolor=magenta,     
	urlcolor=blue         
}

\begin{document}
\title{Palindromic Subsequences in Finite Words}
%
%
\author{Clemens M\"{u}llner$^{1,}$\thanks{This research was suported by the European Research Council (ERC) under the European Union’s Horizon 2020 research and innovation programme under the Grant Agreement No 648132.} \and Andrew Ryzhikov$^2$}

\date{CNRS,  Universit\'e Claude Bernard - Lyon 1, France \textit{mullner@math.univ-lyon1.fr} \\ LIGM, Universit\'e Paris-Est, Marne-la-Vall\'ee, France \textit{ryzhikov.andrew@gmail.com}}
%
%
\maketitle              
\begin{abstract}
In 1999 Lyngs{\o} and Pedersen proposed a conjecture stating that every binary circular word of length $n$ with equal number of zeros and ones has an antipalindromic linear subsequence of length at least~$\frac{2}{3}n$. No progress over a trivial $\frac{1}{2}n$ bound has been achieved since then. We suggest a palindromic counterpart to this conjecture and provide a non-trivial infinite series of circular words which prove the upper bound of~$\frac{2}{3}n$ for both conjectures at the same time.
The construction also works for words over an alphabet of size $k$ and gives rise to a generalization of the conjecture by Lyngs{\o} and Pedersen. Moreover, we discuss some possible strengthenings and weakenings of the named conjectures.
We also propose two similar conjectures for linear words and provide some evidences for them.

\end{abstract}
\section{Introduction}

Investigation of subsequences in words is an important part of string algorithms and combinatorics, with applications to string processing, bioinformatics, error-correcting codes. A lot of research has been done in algorithms and complexity of finding longest common subsequences \cite{Bringmann2015,Abboud2015}, their expected length in random words \cite{Paterson1994}, codes with bounded lengths of pairwise longest common subsequences~\cite{Sloane2008}, etc. An important type of subsequences is a longest palindromic subsequence, which is in fact a longest common subsequence of a word and its reversal. Despite a lot of research in algorithms and statistics of longest common subsequences, the combinatorics of palindromic subsequences is not very well understood. We mention \cite{Bukh2014,Bukh2016,Axenovich2013,Holub2009} as some results in this direction. In this note we recall some known conjectures on this topic and provide a number of new ones.

The main topic of this note are finite words. A {\em linear word} (or just a {\em word}) is a finite sequence of symbols over some alphabet. A {\em subsequence} of a linear word $w = a_1 \ldots a_n$ is a word $w' = a_{i_1} \ldots a_{i_m}$ with $i_1 < \ldots < i_m$. A {\em circular word} is an equivalence class of linear words under rotations. Informally, a circular word is a linear word written on a circle, without any marked beginning or ending. A linear word is a {\em subsequence} of a circular word if it is a subsequence of some linear word from the corresponding equivalence class (such linear word is called a {\em linear representation}).

A word $w = a_1 \ldots a_n$ is a {\em palindrome} if $a_i = a_{n - i + 1}$ for every $1 \le i \le \frac{n}{2}$. A word is called {\em binary} if its alphabet is of size two (in this case we usually assume that the alphabet is $\{0, 1\}$). A binary word $w = a_1 \ldots a_n$ is an {\em antipalindrome} if $a_i \not = a_{n - i + 1}$ for every $1 \le i \le \frac{n}{2}$. The {\em reversal} $w^R$ of a word $w = a_1 \ldots a_n$ is the word $a_n \ldots a_1$.

In 1999 Lyngs{\o} and Pedersen formulated the following conjecture motivated by analysis of an approximation algorithm for a 2D protein folding problem \cite{Lyngso1999}.

\begin{conjecturer}[Lyngs{\o} and Pedersen, 1999]\label{conj-weak-anti}
	Every binary circular word of length $n$ divisible by $6$ with equal number of zeros and ones has an antipalindromic subsequence of length at least $\frac{2}{3}n$.
\end{conjecturer}

To the best of our knowledge, no progress has been achieved in proving this conjecture, even though it has drawn substantial attention from the combinatorics of words community. However, it is a source of other interesting conjectures.

In the mentioned conjecture, the position of a longest antipalindromic subsequence on the circle is arbitrary. A strengthening is to require the two halves of the subsequence to lie on different halves of the circle according to some partition of the circle into two parts of equal length. Surprisingly, experiments show that this does not change the bound.

\begin{conjecturer}[Brevier, Preissmann and Seb\H{o}, \cite{Brevier}] \label{conj-strong-anti}
	Let $w$ be a binary circular word of length $n$ divisible by $6$ with equal number of zeros and ones. Then $w$ can be partitioned into two linear words $w_1, w_2$ of equal length, $w = w_1w_2$, having subsequences $s_1, s_2$ such that $s_1s_2$ is an antipalindrome and $|s_1| = |s_2| = \frac{1}{3}|w|$.
\end{conjecturer}

We checked this conjecture up to $n = 30$ by computer. The worst known case for the both conjectures is provided by the word $w = 0^i 1^i (01)^i 1^i 0^i$ showing the tightness of the conjectured bound (by tightness everywhere in this note we understand the existence of a lower bound different from the conjectured bound by at most a small additive constant). The bound $\frac{1}{2}n$ instead of $\frac{2}{3}n$ can be easily proved, but no better bound is known.

\begin{propositionr}[Brevier, Preissmann and Seb\H{o}, \cite{Brevier}]\label{prop-anti}
	Conjecture \ref{conj-strong-anti} is true when replacing $|s_1|=|s_2|=\frac{1}{3}|w|$ by $|s_1|=|s_2|=\frac{1}{4}|w|$.
\end{propositionr}
\begin{proof}
	Consider an arbitrary partition of $w$ into two linear words $w_1, w_2$ of equal length, $w = w_1w_2$. The number of zeros in $w_1$ is the same as the number of ones in $w_2$ and vice versa. Thus we can pick an antipalindromic subsequence $0^k1^k$ or $1^k0^k$ with $k \ge \frac{1}{4}n$ having the required properties.
\end{proof}



\section{Circular Words}

A natural idea is to look at palindromic subsequences instead of antipalindromic ones. This leads to a number of interesting conjectures which we describe in this section. First, we formulate palindromic counterparts to Conjectures \ref{conj-weak-anti} and \ref{conj-strong-anti}.

\begin{conjecturer}\label{conj-weak-pal}
	Every binary circular word of length $n$ has a palindromic subsequence of length at least $\frac{2}{3}n$.
\end{conjecturer}

\begin{conjecturer}\label{conj-strong-pal}
	Let $w$ be a binary circular word of length $n$ divisible by $6$. Then $w$ can be partitioned into $2$ linear words $w_1, w_2$ of equal length, $w = w_1w_2$, having subsequences $s_1, s_2$ such that $s_1 = s_2^R$ (that is, $s_1s_2$ is a palindrome) and $|s_1s_2| = \frac{2}{3}|w|$.
\end{conjecturer}

We checked both conjectures up to $n = 30$ by computer.
The worst known case for Conjecture \ref{conj-strong-pal} is provided by the word $0^{2i} (10)^i 1^{2i}$, showing the tightness of the conjectured bound. 
The word $0^i (10)^i 1^i$ provides an upper bound of $\frac{3}{4}n$ for Conjecture~\ref{conj-weak-pal}. A better bound is discussed in Section \ref{sect-lb}.

In Conjecture \ref{conj-weak-pal} it is enough to pick the subsequence consisting of all appearances of the letter with the largest frequency to get the $\frac{1}{2}n$ lower bound. It is also easy to prove the $\frac{1}{2}n$ bound for Conjecture \ref{conj-strong-pal}. No better bounds are known to be proved.

\begin{propositionr}
	Conjecture \ref{conj-strong-pal} is true when replacing $|s_1s_2| = \frac{2}{3}|w|$ by $|s_1s_2| = \frac{1}{2}|w|$.
\end{propositionr}
\begin{proof}
	Consider an arbitrary partition of $w$ into two linear words $w_1, w_2$ of equal length, $w = w_1w_2$. Assume that the number of ones in $w$ is at least $\frac{|w|}{2}$, and $w_1$ has less ones than $w_2$. By changing the partition by one letter each time (by adding a subsequent letter to the end of $w_1$ and removing one from the beginning), we get an opposite situation in $\frac{1}{2}n$ steps. That means that there exists a partition $w = w'_1w'_2$, $|w'_1| = |w'_2|$, such that the number of ones in $w'_1$ is the same as the number of ones in $w'_2$. Thus, we can pick an antipalindromic subsequence $1^k$ with $k \ge \frac{1}{2}n$ having the required properties.
\end{proof}

Conjecture \ref{conj-strong-pal} is about a palindromic subsequence aligned with some cut of the circular word into two equal halves. There are $\frac{n}{2}$ such cuts, so one attempt to simplify the conjecture is to look at only two cuts which are ``orthogonal''. This way we attempt to switch from the circular case to something close to the linear case, which is often easier to deal with.

Let $w$ be a circular word of length $n$ divisible by $4$. Let $w_1 w_2 w_3 w_4$ be some partition of $w$ into four linear words of equal length. Let $p_1p'_1$ and $p_2p'_2$, $|p_1| = |p'_1|$, $|p_2| = |p'_2|$, be the longest palindromic subsequences of $w$ such that $p_1$, $p'_1$, $p_2$, $p'_2$ are subsequences of $w_1w_2$, $w_3w_4$, $w_2w_3$, $w_4w_1$ respectively. Informally, these two palindromes are aligned to two orthogonal cuts of the word $w$ into two linear words of equal length. The partitions $w_1w_2, w_3w_4$ and $w_2w_3, w_4w_1$ are two particular partitions (made by two orthogonal cuts) considered among all $\frac{n}{2}$ partitions in Conjecture \ref{conj-strong-pal}.

\begin{conjecturer}\label{conj-cuts}
	For every word $w$ of length $n$ divisible by $4$ and its every linear representation $w = w_1 w_2 w_3 w_4$, the maximum of the lengths of $p_1p'_1$ and $p_2p'_2$ defined above is at least $\frac{1}{2}n$.
\end{conjecturer}

We checked this conjecture up to $n = 30$ by computer. The worst known case is provided by the already appeared word $0^i (10)^i 1^i$ showing the tightness of the conjectured bound. The bound $\frac{1}{3}n$ can be proved as follows.

\begin{propositionr}
	For every word $w$ of length $n$ divisible by $4$ and its every linear representation $w = w_1 w_2 w_3 w_4$, the maximum of the lengths of $p_1p'_1$ and $p_2p'_2$ is at least $\frac{1}{3}n$.
\end{propositionr}
\begin{proof}
	Suppose that $|p_1p'_1| < \frac{1}{3}n$. Then without loss of generality we can assume that the number of zeros in $w_1w_2$ and the number of ones in $w_3w_4$ is less than~$\frac{1}{6}n$. Then by the pigeonhole principle the number of ones in both $w_1$ and $w_2$, and the number of zeros in both $w_3$ and $w_4$ is at least $\frac{1}{12}n$. It means that we can pick a subsequence of $\frac{1}{12}n$ zeros and then $\frac{1}{12}n$ ones from $w_4w_1$ and a symmetrical subsequence from $w_2w_3$. Thus we get $|p_2p'_2| \ge \frac{1}{3}n$.
 \end{proof}

In fact, a slightly stronger statement that the total length of $p_1p'_1$ and $p_2p'_2$ is $\frac{2}{3}n$ can be proved this way. We conjecture the optimal bound for this value to be equal to $n$.

Even being proved, the bound of $\frac{1}{2}n$ in this conjecture would not improve the known bound for Conjecture \ref{conj-strong-pal}. However, Conjecture \ref{conj-cuts} deals with palindromic subsequences of only two linear words, and thus seems to be easier to handle. Considering four regular cuts instead of two should already improve the bound for Conjecture \ref{conj-strong-pal}.

\section{Showing Asymptotic Tightness of Conjecture \ref{conj-weak-pal}} \label{sect-lb}

In this section we present the main technical contribution of this paper, which is an infinite family of words providing a better upper bound for Conjecture \ref{conj-weak-pal}. In fact, we show a stronger result for words over an arbitrary alphabet. Below we consider words over the alphabet $\{0,\ldots,k-1\}$, i.e. $w \in \{0,\ldots, k-1\}^{*}$.

\begin{definitionr}
	We say that $w'$ is a {\em consecutive subword} of a word $w$ if there exist words $u, v$ with $w = uw'v$.
	
	We call a word $w \in \{0,\ldots, k-1\}^{*}$ \emph{of type $n$} if it is a consecutive subword of $(0^{n} 1^{n} \ldots (k-1)^{n})^{*}$ or a consecutive subword of $((k-1)^{n} \ldots 1^{n} 0^{n})^{*}$. In the first case we write $w \in S_n'$, in the second case we write $w \in S_n''$.
	 
	Furthermore, we define $S_n = S_n' \cup S_n''$.
\end{definitionr}

Thus $w \in S_n'$ if it is a concatenation of blocks $(0^{n} 1^{n} \ldots (k-1)^{n})$, where the first and the last blocks may be shorter, and analogously for $w \in S_n''$.

We denote by $\overline{w}$ the word we get when exchanging every letter $\ell$ by $(k-1-\ell)$, e.g. $\overline{01\ldots(k-1)} = (k-1)(k-2)\ldots0$.
We see directly that $w \in S_n$ if and only if $\overline{w} \in S_n$.
Furthermore, we have that $w \in S_n$ if and only if $w^{R}\in S_n$.

\begin{lemmar}\label{le:upperBound}
	Let $w_1 \in S_{n_1}$ be a word of length $\ell_1$ and $w_2 \in S_{n_2}$ be a word of length $\ell_2$, where $n_1 > n_2$.
	Then, the length of the longest common subsequence of $w_1$ and $w_2$ is at most $\frac{\ell_1 + \ell_2}{k+1} + \ell_1 \frac{n_2}{n_1} + 2 n_2$.
\end{lemmar}
\begin{proof}
	Let $w$ be a common subsequence of $w_1$ and $w_2$ of length $\ell$. 
	We see that~$w$ is of the form $a_1^{p_1} a_2^{p_2} \ldots a_s^{p_s}$, where all $a_j \in \{0,\ldots,k-1\}$, all $p_j$ are positive and $a_j \neq a_{j+1}$.
	We find directly that for $i = 1,2$: 
	
	\begin{align}\label{eq:boundS}
	s \leq \ceil{\frac{\ell_i -1}{n_i}} + 1 \leq \frac{\ell_i - 1 + n_i -1}{n_i} +1 \leq \frac{\ell_i}{n_i} +2.
	\end{align}
	
	We consider now the minimal length of a consecutive subword of $w_i$ that contains $a_j^{p_j}$, where $p_j > n_i$. Thus, $a_j^{p_j}$ cannot be contained in one block of the form $(0^{n_i}1^{n_i}\ldots(k-1)^{n_i})$.
	This shows that the minimal length of a consecutive subword of $w_i$ that contains $a_j^{p_j}$ is at least $kn_i$.
	
	This generalizes for $p_j > n_ir$ and we find that each $a_j^{p_j}$ spans a subsequence of length at least $k n_i (\ceil{\frac{p_j}{n_i}}-1) \geq k (p_j - n_i)$ in $w_i$. Thus, we find $\ell_i \geq \sum_{j=1}^{s} k (p_j - n_i)$. This gives in total
	\begin{align}\label{eq:boundL}
	\ell = \sum_{j=1}^{s} p_j \leq \frac{\ell_i}{k} + s n_i.
	\end{align}
	By combining \eqref{eq:boundS} and \eqref{eq:boundL} we find
	\begin{align}
	\ell \leq \frac{\ell_2}{k} + (\frac{\ell_1}{n_1} + 2) n_2.
	\end{align}
	Furthermore, we find directly that $\ell \leq \ell_1$.
	This gives in total
	\begin{align*}
	\frac{k}{k+1} \ell &\leq \frac{\ell_2}{k+1} + \ell_1 \frac{k n_2}{(k+1) n_1} + \frac{2k}{k+1} n_2 \leq \frac{\ell_2}{k+1} + \ell_1 \frac{n_2}{n_1} + 2 n_2\\
	\frac{1}{k+1} \ell &\leq \frac{\ell_1}{k+1},
	\end{align*}
	and by adding these inequalities, we find
	\begin{align*}
	\ell \leq \frac{\ell_1 + \ell_2}{k+1} + \ell_1 \frac{n_2}{n_1} + 2 n_2.
	\end{align*}
 \end{proof}

We think of $\frac{\ell_1 + \ell_2}{k+1}$ in the bound above as the ``main term''. Therefore, we need that $\frac{n_2}{n_1}$ is small.
The remaining term origins from boundary phenomena due to incomplete blocks. 
We note that this ``main term'' is indeed sharp for large $\ell_1, \ell_2$, when $\frac{n_1}{n_2}$ is integer and $k \ell_1 = \ell_2$ as the following example shows.
\begin{exampler}
	We consider $n_1 = p n_2$, with $p$ integer, and $w_1 = (0^{n_1} 1^{n_1}\ldots (k-1)^{n_1})^{\ell n_2}, w_2 = (0^{n_2} 1^{n_2}\ldots (k-1)^{n_2})^{k \ell n_1} = ((0^{n_2} 1^{n_2}\ldots (k-1)^{n_2})^{k p})^{\ell n_2}$. 
	One finds that $i^{n_1}$ is a subsequence of $(0^{n_2} 1^{n_2})^{p}$ and thus, $w_1$ is a subsequence of $w_2$.  
	This gives directly $|w_1| = k n_1 \ell n_2, |w_2| = k n_2 k \ell n_1 = k|w_1|$ and $|w| = |w_1| = \frac{|w_1| + |w_2|}{k+1}$.
\end{exampler}

For the following considerations we will need a generalization of the notion of antipalindromes to the case of non-binary alphabet. One natural version would be to say that $w$ is an antipalindrome if $w$ and $w^{R}$ differ at every position. However, we work with a stronger notion, which still provides an interesting bound. 

\begin{definitionr}
	We call a word $w \in \{0,\ldots,k-1\}^{*}$ a {\em strong antipalindrome} if $w = \overline{w}^{R}$.
\end{definitionr}

\begin{theorem}\label{thm-main-lb}
	For every $\varepsilon > 0$ there exists a circular word over the alphabet $\{0,\ldots,k-1\}$ with equal number of $0$'s, $1$'s, $\ldots$, $(k-1)$'s ($n$ occurences of each letter) such that any palindromic and any strongly antipalindromic subsequence of it is of length at most $(\frac{2}{k + 1} + \varepsilon)kn$.
\end{theorem}
\begin{proof}
	Let us consider a circular word with a linear representation $w_1 w_2 \ldots w_r = w$, where $w_j = (0^{p^j} 1^{p^j}\ldots (k-1)^{p^j})^{p^{r-j}}$. We see directly that $|w_j| = k p^r$ and, thus, $k n := |w| = k r p^r$.
	Furthermore, we have $w_j \in S_{p^j}$.
	
	We only work in the palindromic case from now on, but the same reasoning also holds in the case of strong antipalindromes.
	
	Let $v v^{R}$ be a palindromic subsequence of even length.
	Thus, we find that $v$ is a subsequence of the linear word $u_1' w_{i_1} w_{i_2} \ldots w_{i_a} u_2$ and $v^{R}$ is a subsequence of the linear word $u_2' w_{j_1} w_{j_2} \ldots w_{j_b} u_1$, where $u_1 u_1' = w_{i_0}, u_2 u_2' = w_{j_0}$ and $i_{k} \neq j_{\ell}$ for all $0\leq i \leq a, 0\leq \ell \leq b$.
	
	This shows that $v$ is a common subsequence of $u_1' w_{i_1} w_{i_2} \ldots w_{i_a} u_2$ and\\
	$u_1^{R} w_{j_b}^{R}  \ldots w_{j_1}^{R} u_2'^{R}$.
	By removing the parts of $v$ that belong to the boundary blocks $u_i$ we get $v$ that is a common subsequence of $w_{i_1} w_{i_2} \ldots w_{i_a}$ and $w_{j_b}^{R}  \ldots w_{j_1}^{R}$, where
	\begin{align*}
	|v| - |v'| \leq |w_{i_0}| + |w_{j_0}| = 2kp^r.
	\end{align*}
	
	From now on, we only work with $v'$. 
	We can rewrite $v'$ as a concatenation of at most $(a+b-1)$ blocks $v_i$, where each $v_i$ is a common subsequence of some $w^{(i)}_1 \in S_{p^{j_1(i)}}$ and $w^{(i)}_2 \in S_{p^{j_2(i)}}$ where $j_1(i) \neq j_2(i)$.
	Furthermore, we have $a+b \leq r$ and 
	\begin{align*}
	\sum_{i} |w^{(i)}_1| = akp^r\\
	\sum_{i} |w^{(i)}_2| = bkp^r.
	\end{align*}
	
	By using Lemma~\ref{le:upperBound} we find that
	\begin{align*}
	|v'| &= \sum_{i} |v_i|\\ 
	&\leq \sum_{i} \Big( \frac{|w^{(i)}_1| + |w^{(i)}_2|}{k+1} + \frac{(|w^{(i)}_1| + |w^{(i)}_2|)}{p} + 2 p^{r-1} \Big)\\
	&\leq \frac{|w|}{k+1} + \frac{|w|}{p} + \frac{2|w|}{k p}.
	\end{align*}
	
	This gives in total (together with the bound on $|v| - |v'|$)
	\begin{align*}
	|v v^{R}| \leq \frac{2|w|}{k+1} + |w| \rb{\frac{4}{p} + \frac{4}{r}}.
	\end{align*}
	Thus, choosing $p = r \geq \frac{8}{\varepsilon}$ finishes the proof.
 \end{proof}

The trivial lower bound is $\frac{1}{k}$.
For palindromes, this can be seen immediately. For strong antipalindromes the case for $k$ odd works very similarly: We see that $\overline{(k-1)/2} = (k-1)/2$ and the word $((k-1)/2)^{|w|/k}$ is a strongly antipalindromic subsequence of length $|w|/k$. The case $k$ is even slightly more complicated but can be dealt with in the same way as $k=2$.

Theorem \ref{thm-main-lb} deserves some remarks. First, it is interesting that the family of words constructed in the theorem provides the same bound for both palindromic and strongly antipalindromic subsequences. Second, it provides a generalization of the palindromic and strongly antipalindromic conjectures to the case of an alphabet of more than two letters. These conjectures also remain open. 

Finally, for any $\varepsilon>0$, we find that the bound $\frac{2n}{k+1 + \varepsilon}$ holds almost surely for large $n$ in the case when we choose every letter independently and uniformly in $\{0,\ldots,k-1\}$.

To see this, we fix a subsequence of length $\frac{n}{k+1 + \varepsilon}$ and call it $w_0$.
Then we try to find $w_0, \overline{w_0}, w_0^R$ or $\overline{w_0}^{R}$ as a subsequence of the remaining word $w_1$.
However, any letter in $w_1$ is chosen independently and uniformly.
Therefore, it takes on average $k$ letters until one finds one specific letter.
By the law of large numbers, the number of letters we have to read in a string of independent and uniformly chosen letters to find a specific subsequence of length $\ell$ is asymptotically normal distributed with mean $\ell k$ and variance $\alpha \ell$ for some $\alpha >0$. By the Chebyshev inequality, we find that $w_0$ (or any of the mentioned forms above) appears in $w_1$ almost surely for large $n$ as $|w_1| = (k+\varepsilon) |w_0|$.

\section{Linear Words}

The minimum length of the longest palindromic/antipalindromic subsequence in the class of all linear binary words with $n$ letters can be easily computed. However, for some restricted classes of words their behavior is more complicated. One of the simplest restrictions is to forbid some number of consecutive equal letters. The following proposition is then not hard to prove. It suggests some progress for Conjectures \ref{conj-weak-anti} and \ref{conj-weak-pal} for binary words without three consecutive equal letters.

\begin{propositionr}
	Every binary word of length $n$ without three consecutive equal letters has a palindromic subsequence of length at least $\frac{2}{3}(n - 2)$. The same is true for an antipalindromic subsequence.
\end{propositionr}
\begin{proof}
	Let $w$ be a binary word without three consecutive equal letters. Consider the representation $w = w_1 w_2 \ldots w_m$ such that each $w_i$ is composed of only zeros or only ones, and two consecutive words $w_i$ and $w_{i + 1}$ consist of different letters. Then the length of each $w_i$ is at most $2$. Assume that $m$ is even (otherwise remove~$w_m$). Then one can pick at least one letter from each pair $w_i$, $w_{m - i + 1}$ (or two letters if both $w_i$, $w_{m - i + 1}$ are of the same length) and all the letters from~$w_{\frac{m + 1}{2}}$ in such a way that the resulting subsequence is a palindrome. This way we get a palindromic subsequence of length at least $\frac{2}{3}(n - 2)$. The same proof can be done for antipalindromic subsequences.
 \end{proof}

For the antipalindromic part, one can take the word $(001)^i$ to see tightness (we conjecture the bound $\frac{2}{3}n$ to be tight for words with equal number of zeros and ones, but we could not find an example providing tightness). For palindromic subsequences we conjecture a stronger bound.

\begin{conjecturer}
	Every binary word of length $n$ without three consecutive equal letters has a palindromic subsequence of length at least $\frac{3}{4}(n - 2)$.
\end{conjecturer}

We checked this conjecture up to $n = 30$. The worst known cases are provided by the word $(001)^i (011)^i$, showing the tightness of the conjectured bound.

Note that every binary word without two consecutive equal letters is a sequence of alternating zeros and ones, and thus has a palindromic subsequence of length $n-1$, where $n$ is the length of the word. For a three-letter alphabet it is not hard to prove the following.

\begin{propositionr}
	Let $w$ be a word of length $n$ over a three-letter alphabet. If $w$ has no two consecutive equal letters, then it has a palindromic subsequence of length at least $\frac{1}{2}(n - 1)$.
\end{propositionr}
\begin{proof}
	Assume that the number of letters in $w$ is even (otherwise, remove the last letter). Let $w = w_1 w_2 \ldots w_m$ where $w_i$ is a word of length $2$. Each such word contains two different letters. Then for each pair $w_i$, $w_{m - i + 1}$ there exists a letter present in both words. By taking such a letter from every pair, we get a palindrome of length $m = \frac{1}{2}(n - 1)$.
 \end{proof}

Based on these observations and computer experiments, we formulate the following conjecture.

\begin{conjecturer}
	Let $w$ be a word of length $n$ over an alphabet of size $k$, $k \ge 2$. If $w$ has no two consecutive equal letters, then it has a palindromic subsequence of length at least $\frac{1}{k - 1}(n - 1)$.
\end{conjecturer}

We checked this conjecture up to $n = 21$ for $k = 4$ and $n = 18$ for $k = 5$ by computer. A critical example for this conjecture is provided by a word which is a concatenation of the word $(a_1 a_2)^i$ and words $(a_{\ell + 1} a_\ell)^{i - 1} a_{\ell + 1}$ for $1 < \ell < k - 1$. This word shows that the conjectured bound is tight.

\section{Further Work}

There are some questions besides the conjectures above that are worth mentioning. First, there is no known reduction between the palindromic and antipalindromic conjectures. Thus, it is interesting to know whether a bound for one of them implies some bound for the other one. Second, no non-trivial relation is known for the bounds for the same conjecture but different size of alphabets. 

\section*{Acknowledgements}

We thank anonymous reviewers for their comments on the presentation of the paper. The second author is also grateful to Andr{\'a}s Seb\H{o}, Michel Rigo and Dominique Perrin for many useful discussions during the course of the work.

%
%
 \bibliographystyle{splncs04}
 \bibliography{words}

\end{document}